\documentclass[12pt]{article}
\usepackage[T1]{fontenc}
 \usepackage{stix}
 \usepackage{libertine,libertinust1math}
\usepackage[dvipsnames]{xcolor}
\usepackage{pdfpages}
\usepackage{sgame}
\usepackage{accents}
\usepackage{tikz-cd}
\usepackage{float}
\usepackage[round]{natbib}   
\usepackage{multirow}
\usepackage{dutchcal}
\usepackage{pdfpages}
\usepackage{sgame}
\usepackage{mathtools}
\usepackage{amsthm}
\usepackage{enumitem}
\usepackage{epigraph}
\usetikzlibrary{calc}
\usetikzlibrary{arrows}
\usepackage{amsmath}

\newtheorem{theorem}{Theorem}[section]
\newtheorem{proposition}[theorem]{Proposition}

\newtheorem{corollary}[theorem]{Corollary}

\theoremstyle{definition}

\newtheorem{definition}[theorem]{Definition}
\newtheorem{example}[theorem]{Example}

\usepackage{hyperref}

\providecommand{\keywords}[1]{\textbf{\textit{Keywords:}} #1}
\providecommand{\jel}[1]{\textbf{\textit{JEL Classifications:}} #1}
\let\vec\mathbf


\begin{document}
\author{Joseph Whitmeyer\thanks{Department of Sociology, University of North Carolina at Charlotte} \and Mark Whitmeyer\thanks{Institute for Microeconomics and Hausdorff Center for Mathematics, University of Bonn \newline Email: \href{mailto:mark.whitmeyer@gmail.com}{mark.whitmeyer@gmail.com}}}
\title{Mixtures of Mean-Preserving Contractions\footnote{We thank Ying Chen and two anonymous referees for their useful suggestions. We are also grateful to Gill Grindstaff and Jan Schlupp for their help and advice. Mark Whitmeyer's work was generously funded by the Deutsche Forschungsgemeinschaft (DFG, German Research Foundation) under Germany's Excellence Strategy-GZ 2047/1, Projekt-ID 390685813.}}
\date{\today{}}
\maketitle
\begin{abstract}
Given any purely atomic probability distribution with support on $n$ points, $P$, any mean-preserving contraction (mpc) of $P$, $Q$, with support on $m > n$ points is a mixture of mpcs of $P$, each with support on at most $n$ points. We illustrate several applications of this result to Bayesian persuasion and information design.
\end{abstract}
\\
\jel{C72; D82; D83}

\section{Introduction}
The mean-preserving contraction (mpc) of a probability distribution, in which a probability distribution is altered by collapsing portions of its measure to their respective barycenters, is an important concept in economics. In recent years there has been a surge of papers in which information structures are endogenous, and many works involve optimization problems in which one or many agents choose from the set of mpcs of a given distribution. 

The purpose of this paper is to establish a result that is useful both in clarifying the economic intuition behind the solutions to, as well as simplifying the solving of, such problems. Our main finding, Theorem \ref{main}, establishes that given a purely atomic probability measure $P$ on $\mathbb{R}$ with support on $n$ points, any mpc of $P$ with support on $n+1$ points, $Q$, is the convex combination of two mpcs of $P$, $Q'$ and $Q''$, each of which have support on at most $n$ points. An important corollary of this result, Corollary \ref{useful}, is that any mpc, $Q$, of $P$ is a mixture of mpcs with support on at most $n$ points. 

These discoveries imply several compelling results related to Bayesian persuasion and information design,\footnote{The number of papers in this area has grown rapidly in the decade following the writing of \cite{rayo} and \cite{kam}, the two papers that first sparked the widespread interest in the topic. \cite{mor} provides a recent and relatively self-contained overview of the topic and \cite{kam2} provides a comprehensive survey of the literature.} a literature that studies how a principal (or principals) can design information structures in various environments in order to achieve some objective. One particularly tractable class of problems are those in which the state is a random variable and the sender and receiver's utilities are linear in the state (see e.g. \cite{gent,kol,kol2}; and \cite{martini}). Corollary \ref{useful} allows for an elementary proof for the upper bound of the number of messages needed by the persuader in an optimal mechanism in such problems (Proposition \ref{bound}).

Corollary \ref{useful} also has useful ramifications for competitive Bayesian persuasion problems, in which multiple principals compete by designing information structures. Our result implies that when constructing equilibria, instead of checking deviations to \textit{any} mpc of $P$, one need check only deviations to mpcs with support on $n$ points. In Subsection \ref{example} we present a simple competitive persuasion problem in which this result is used, and \cite{whit3} and \cite{jain} both appeal to Corollary \ref{useful} in similar settings to this example. The existence of a symmetric mixed strategy equilibrium with support on $n$ (or fewer)-point mpcs (Proposition \ref{mix}) also follows from this corollary.

This paper contributes to the literature in economics and mathematics on statistical experiments and majorization. The first foray into the area is the seminal work of \cite{hard29} (see also \cite{hard}), who establish a fundamental result on majorization, which was shortly followed and generalized by \cite{black} and \cite{strass}.\footnote{Another related paper is \cite{elt}, who provide a constructive proof of \cite{hard29}'s main result and show how it can be generalized to infinite-dimensional spaces.} The equivalent notion of a mean-preserving spread (mps)--the opposite of an mpc in which portions of a probability measure are ``spread'' out--was subsequently introduced by \cite{rot}. 


More recently, \cite{hill} introduce the concept of a ``fusion'' of a probability distribution. A fusion, $Q$, of a distribution $P$ is any distribution that can be obtained by collapsing parts of the mass of $P$ to their respective barycenters. In the setting we analyze in this paper, where $P$ is a measure on $\mathbb{R}$ (with finite first moment), the notions of fusion and mean-preserving contraction are equivalent. A key contribution of \cite{hill} is the construction of the set of fusions: just as a measurable function can be approximated by simple functions, which are themselves finite linear combinations of indicator functions, a fusion can be approximated by simple fusions, which are themselves finite compositions of elementary fusions.\footnote{In their parlance, an elementary fusion is just the result of collapsing part of a single (Borel) subset of the probability measure to the barycenter of the subset.} 

In this spirit, the contribution of our paper can be seen as it providing an alternate characterization for the set of fusions of a (finitely supported) probability measure $P$ (on $\mathbb{R}$), one that we believe should be particularly compelling (and useful) to economists. Rather than taking elementary fusions as the building blocks, as \cite{hill} do, our basic objects are the fusions of $P$ with support on $n$ points. Consequently, Corollary \ref{useful} implies that any fusion can be approximated by finite \textit{mixtures} of fusions with $n$-point support.

As we discuss later on in the paper, the set of all mpcs of a discrete probability measure is a compact and convex set. By the Krein-Milman theorem, such a set is merely the closed and convex hull of its extreme points. This is valuable in well-behaved linear optimization problems, in which the objective function obtains a maximum at an extreme point of the constraint set. This theorem (Krein-Milman) and its extension, Choquet's theorem, provide part of the impetus behind the remarkable recent paper by \cite{strack}, who characterize the extreme points of the set of monotonic functions that majorize (or are majorized by) a given monotonic function such as, e.g., a cumulative distribution function (cdf). They thus characterize the set of mpcs and mpss for an arbitrary distribution. 

We believe that our paper can be seen as complementary to \cite{strack}. Indeed their main result in which they characterize the extreme points of the set of mpcs of a random variable with cdf $F$--Theorem 2--requires that $F$ be continuous. In turn, Corollary \ref{useful} in this paper provides a necessary condition for a distribution to be an extreme point of the set of mpcs of a distribution with $n$-point support--it must have support on $n$ points or fewer. 


\subsection{The Main Result}

Throughout, bold font identifies vectors. Let $X = \mathbb{R}$ and denote by $\mathcal{B}$ the Borel subsets of $X$. Moreover, let $\mathcal{P}$ denote the set of Borel probability measures on $\left(X,\mathcal{B}\right)$. Let $P \in \mathcal{P}$ be a purely atomic probability measure with support on $n$ points i.e. $supp P = \vec{a} \coloneqq \left\{a_{1},a_{2},\dots,a_{n}\right\}$,  with respective masses $p_{1}, p_{2}, \dots, p_{n}$, where $p_{i}>0$ for all $i = 1,\dots, n$ and $\sum_{i}^{n}p_{i} = 1$. Denote $\Vec{p} \coloneqq (p_{1}, p_{2}, \dots, p_{n})$. Without loss of generality, $a_{1} < a_{2} < \cdots < a_{n}$. Let $Q \in \mathcal{P}$ be a purely atomic probability measure with support on $m$ points: $supp Q = \vec{b} \coloneqq \left\{b_{1},b_{2},\dots,b_{m}\right\}$,  with respective masses $q_{1}, q_{2}, \dots, q_{m}$, where $q_{i}>0$ for all $i = 1,\dots, m$ and $\sum_{i}^{m}q_{i} = 1$. Denote $\vec{q} \coloneqq (q_{1}, q_{2}, \dots, q_{m})$. Finally, denote $\vec{pa} \coloneqq (p_{1}a_{1}, p_{2}a_{2}, \dots, p_{n}a_{n})$ and analogously for $\vec{qb}$.

We adopt the following definition from \cite{hill} and \cite{elt} for simple fusions:\footnote{Recall that the notions of a fusion and an mpc of a probability measure $P$ are equivalent.}

\begin{definition}\label{def1}
$Q$ is a \hypertarget{fs}{\textbf{Simple Mean-Preserving Contraction (SMPC)}} of $P$ if there exists a non-negative row-stochastic (Markov) $n \times m$ matrix $F$ that satisfies
\[\Vec{p}F = \vec{q}\]
and
\[\left(\Vec{pa}\right)F = \vec{qb}\]
Denote by $\mathcal{m}\left(P\right)$ the set of all smpcs of $P$.
\end{definition}
Consequently, an smpc with $n+1$ atoms, $\vec{b}$, of an arbitrary $n$-atom probability measure $P$ with atoms, $\vec{a}$, can be defined by the $n \times (n+1)$ Markov matrix $F$. Each row of $F$, $F_i$, sums to $1$ and denotes the partition of the weight $P_i$ at atom $a_i$ across the atoms $\vec{b}$.

To better grasp the concept (and notation) consider the following example:
\begin{example}\label{ex1}
Let $P$ be a discrete probability distribution with support on three points:
\[P = \begin{Bmatrix} 
0  & \frac{1}{2} & 1 \\
\frac{3}{10} & \frac{3}{10} &  \frac{2}{5}
\end{Bmatrix}\]
where the top row is the vector of support, and the bottom row is the vector of associated probability weights. Then, $Q$ is a smpc of $P$, with an associated Markov matrix, $F$:
\[Q = \begin{Bmatrix} 
\frac{1}{6} & \frac{1}{2} & \frac{3}{4} & \frac{5}{6} \\
\frac{3}{10} & \frac{1}{5} & \frac{1}{5} & \frac{3}{10}
\end{Bmatrix} \quad \text{\&} \quad F = \begin{pmatrix}
\frac{2}{3} & \frac{1}{3} & 0 & 0\\
\frac{1}{3} & 0 & \frac{1}{3} & \frac{1}{3}\\
0 & \frac{1}{4} & \frac{1}{4} & \frac{1}{2}
\end{pmatrix}\]
\end{example}

After \cite{hill}, we construct the set of mean-preserving contractions of $P$, $\mathcal{M}\left(P\right)$, by taking the weak-* closure of $\mathcal{m}\left(P\right)$:

\begin{definition}\label{def2}
$R \in \mathcal{P}$ is a \hypertarget{fss}{\textbf{Mean-Preserving Contraction (MPC)}} of $P$ if there exists a sequence of smpcs $\left\{Q_{m}\right\}_{m=1}^{\infty} \subset \mathcal{m}\left(P\right)$ that satisfies $Q_{m} \to_{w} R$.\footnote{$\to_{w}$ denotes the standard idea of weak convergence.} Denote by $\mathcal{M}\left(P\right)$ the set of all mpcs of $P$.
\end{definition}

Next, we say that that $Q = \alpha Q' + (1-\alpha)Q''$ for smpcs $Q$, $Q'$, and $Q''$ if and only if \[\label{1529}\tag{$1$}F = \alpha F' + (1-\alpha)F''\] for Markov matrices $F$, $F'$, and $F''$, where $0$-columns are inserted into the Markov matrices corresponding to the zero-probability atoms in $Q'$ and $Q''$ but otherwise the ratios within columns in $F'$ and $F''$ are the same as in $F$.
Using this, we may now state the main theorem: 

\begin{theorem}\label{main}
Let $Q \in \mathcal{m}\left(P\right)$ be any smpc of $P$ with support on $n+1$ points, $n \geq 2$. Then $Q$ is the convex combination of two purely atomic probability measures $Q'$ and $Q''$; $Q', Q'' \in \mathcal{m}\left(P\right)$, each with support on at most $n$ points. $Q'$ and $Q''$ are unique.
\end{theorem}

The essence of the proof is as follows: consider the Markov matrix $F$ corresponding to a smpc $Q$ where $m = n+1$. Because $F$ is an $n \times (n+1)$ matrix, there is some subset of column vectors that is linearly dependent, and because all atoms are distinct, this subset has a minimal size of three. Hence, any of these column vectors can be expressed as a linear sum of the others. This means we can write a new matrix with any of these column vectors redistributed to the other vectors, such that the horizontal sum across vectors remains the same ($1$ for the whole matrix), the chosen vector is reduced to the $0$-vector, and each other column vector is either the $0$-vector or is a multiple of its original self, i.e., preserves its internal ratios. We refer to this procedure as ``zeroing'' a vector of the matrix. 

Thus, in the proof it suffices to show that we can always choose two different column vectors from the linearly dependent set such that zeroing either of them gives entries in the remaining vectors from $0$ to $1$. We obtain the uniqueness of $Q'$ and $Q''$ by showing that there is no third such vector that can be zeroed independently of the other two, that is, without also zeroing one of the other two vectors.

\begin{proof}
Let $\vec{0}$ and $\vec{1}$ denote the column vectors of all zeroes and ones, respectively. Without loss of generality, suppose the linearly dependent set consists of the entire $n \times n+1$ matrix $F$, that is, $n+1$ column vectors $\vec{f_{j}}$, where $\vec{0} \leq \vec{f_{j}} \leq \vec{1}$ and $\sum_{j=1}^{n+1}f_{ij} = 1$ for all $i$ or $\sum_{j=1}^{n+1}\vec{f_{j}} = \Vec{1}$, the component-wise sum of the vectors $\vec{f_{j}}$. Then there are scalars $\hat{c}_{j}$, not all $0$, such that $\sum_{j=1}^{n+1}\hat{c}_{j}\vec{f_{j}} = \vec{0}$. Because the vectors are non-negative, some of the $\hat{c}_{j}$ will be positive and some negative. Re-index the vectors so that \[\sum_{j=1}^{k}-\left|\hat{c}_{j}\right|\vec{f_{j}} + \sum_{j=k+1}^{n+1}\hat{c}_{j}\vec{f_{j}} = 0\]
for some $k \leq n$. Set $c_{j} = \left|\hat{c}_j\right|$ for all $j$. Then, \[\sum_{j=1}^{k}c_{j}\vec{f_{j}} = \sum_{j=k+1}^{n+1}c_{j}\vec{f_{j}}\]
where $c_{j} \geq 0$ for all $j$. Define \[c_{j^{*}} \coloneqq \max_{1 \leq j \leq k}\left\{c_{j}\right\}, \quad \text{\&} \quad c_{j^{**}} \coloneqq \max_{k+1 \leq j \leq n+1}\left\{c_{j}\right\}\]
Then, the corresponding $\vec{f_{j^{*}}}$ and $\vec{f_{j^{**}}}$ can be zeroed. To see that take one, say $\Vec{f_{j^{*}}}$:
\[\vec{f_{j^{*}}} = -\sum_{\substack{j \leq k\\ j \neq j^{*}}}\frac{c_{j}}{c_{j^{*}}}\vec{f_{j}} + \sum_{j \geq k+1}\frac{c_{j}}{c_{j^{*}}}\vec{f_{j}}\]
In the new Markov matrix for the new mpc, set \[\vec{f_{j}'} = \begin{cases} \left(1 - \frac{c_{j}}{c_{j^{*}}}\right)\vec{f_{j}}, & \quad j \leq k, j \neq j^{*}\\ 
\left(1 + \frac{c_{j}}{c_{j^{*}}}\right)\vec{f_{j}}, & \quad j \geq k+1\\
\vec{0}, & \quad j = j^{*}
\end{cases}\]
We have produced matrix $F'$ by zeroing vector $\vec{f_{j^{*}}}$. By construction, for $j \leq k$, $j \neq j^{*}$:
\[0 \leq \frac{c_{j}}{c_{j^{*}}} \leq 1\]
and so \[0 \leq \left(1 - \frac{c_{j}}{c_{j^{*}}}\right)f_{ij} \leq 1, \quad \text{for all} \quad i, j\]
Also, \[\vec{0} \leq \sum_{\substack{j \leq k\\ j \neq j^{*}}}\left(1 - \frac{c_{j}}{c_{j^{*}}}\right)\vec{f_{j}}<\sum_{\substack{j \leq k\\ j \neq j^{*}}}\vec{f_{j}}<\vec{1}\]
Then  \[\vec{1} = \sum_{j=1}^{n+1}\vec{f_{j}} = \sum_{j=1}^{n+1}\vec{f_{j}'} = \sum_{\substack{j \leq k\\ j \neq j^{*}}}\left(1 - \frac{c_{j}}{c_{j^{*}}}\right)\vec{f_{j}} + \sum_{\substack{j \geq k+1}}\left(1 + \frac{c_{j}}{c_{j^{*}}}\right)\vec{f_{j}}\]
This means \[\sum_{\substack{j \geq k+1}}\left(1 + \frac{c_{j}}{c_{j^{*}}}\right)\vec{f_{j}} \leq \vec{1}\]
i.e. for all $i$, \[\sum_{j \geq k+1}\left(1+\frac{c_{j}}{c_{j}^{*}}\right)f_{ij} \leq 1\]
and since \[\left(1 + \frac{c_{j}}{c_{j^{*}}}\right) \geq 1\]
for all $j$, we have
\[0 \leq \left(1 + \frac{c_{j}}{c_{j^{*}}}\right)f_{ij} \leq 1\]

As a result, for the two zeroed column vectors we create the two distinct $n \times n+1$ Markov matrixes, $F'$ and $F''$, containing different zero vectors, required for Equation \ref{1529}. From the linear dependence the coefficient $\alpha$ in Equation \ref{1529} can be calculated. $Q'$ and $Q''$ are given, of course, by the partitions obtained by removing the zero column vectors from $F'$ and $F''$.

To see that $F'$ and $F''$, and therefore $Q'$ and $Q''$, are unique, suppose for the sake of contradiction that independently of $\vec{f_{j^{*}}}$ and $\vec{f_{j^{**}}}$ there is a third $\vec{f_{j^{\dagger}}}$ that can be zeroed, where $j^{*} \neq j^{\dagger} \neq j^{**}$. Without loss of generality suppose $1 \leq j^{\dagger} \leq k$. Then by definition $c_{j^{\dagger}} \leq c_{j^{*}}$. Without loss of generality let $c_{j^{\dagger}} < c_{j^{*}}$ (in the case of an equality zeroing $\vec{f_{j^{*}}}$ would also zero $\Vec{f_{j^{\dagger}}}$, giving $F'$ and $F''$ again). Thus,
\[\vec{f_{j^{\dagger}}} = -\sum_{\substack{j \leq k\\ j \neq j^{\dagger}}}\frac{c_{j}}{c_{j^{\dagger}}}\vec{f_{j}} + \sum_{\substack{j \geq k+1}}\frac{c_{j}}{c_{j^{\dagger}}}\Vec{f_{j}}\]
Then, because \[\frac{c_{j^{*}}}{c_{j^{\dagger}}} > 1\]
the other zeroing produces the following component, \[f_{ij^{*}}'' = \left(1- \frac{c_{j^{*}}}{c_{j^{\dagger}}}\right)f_{ij^{*}} < 0\]
which is not permitted.
\end{proof}

Next, let us revisit the distribution $P$ and smpc $Q$ from Example \ref{ex1} and the two smpcs, $Q'$ and $Q''$, whose convex combination yields $Q$:

\begin{example}
Recall $P$, $Q$, and $F$ from Example \ref{ex1}. Then,
\[F' = \begin{pmatrix}
\frac{5}{6} & \frac{1}{6} & 0 & 0\\
\frac{5}{12} & 0 & 0 & \frac{7}{12}\\
0 & \frac{1}{8} & 0 & \frac{7}{8}
\end{pmatrix}, \quad \text{\&} \quad F'' = \begin{pmatrix}
\frac{4}{9} & \frac{5}{9} & 0 & 0\\
\frac{2}{9} & 0 & \frac{7}{9} & 0\\
0 & \frac{5}{12} & \frac{7}{12} & 0
\end{pmatrix} \] Thus, 
\[Q' = \begin{Bmatrix} 
\frac{1}{6}  & \frac{1}{2} & \frac{5}{6} \\
\frac{3}{8} & \frac{1}{10} & \frac{21}{40}
\end{Bmatrix}, \quad \text{\&} \quad  Q'' = \begin{Bmatrix} 
 \frac{1}{6} & \frac{1}{2} & \frac{3}{64} \\
 \frac{1}{5} & \frac{1}{3} & \frac{7}{15}
\end{Bmatrix}\]
and $\alpha = 4/7$.
\end{example}

\begin{corollary}\label{useful}
Any mpc $Q \in \mathcal{M}\left(P\right)$ is a mixture of smpcs with support on at most $n$ points. 
\end{corollary}

\begin{proof}
From Theorem \ref{main}, given a purely atomic probability measure $P$ with support on $n$ points, any smpc $Q \in \mathcal{m}\left(P\right)$ with support on $n+1$ points is a convex combination of two smpcs $Q', Q'' \in \mathcal{m}\left(P\right)$ with support on at most $n$ points. If $m = n+1$ then the result is simply Theorem \ref{main}. If not, simply iterate backward until the desired mixture of smpcs of $P$ with support on $n$ points is obtained.

Taking weak limits, this result holds for mpcs of $P$ that are not purely atomic.
\end{proof}


\section{Applications}

Here we illustrate the usefulness of Corollary \ref{useful} in Bayesian persuasion problems. The first subsection illustrates how this result yields easily an upper bound for the number of messages required for the optimal mechanism in a (linear) persuasion problem. The second subsection contains a pair of results related to competitive persuasion. There, we observe that Corollary \ref{useful} simplifies the task of solving for an equilibrium and aids us in the construction of mixed strategy equilibria in such games.

\subsection{A Bound for Linear Persuasion Problems}

One particularly tractable class of persuasion problems are those in which the state is a real-valued random variable and the sender's and receiver's optimal actions are linear in the state , and hence depend only on the expected state (the posterior mean). It is a standard result (see, e.g., \cite{martini}) that this allows us to simplify the persuader's problem to one of choosing any distribution of posterior means that is an mpc of the prior (or equivalently, any cdf over values that second-order stochastically dominates the prior). 

In such an environment, given prior $P$, the sender's persuasion problem is thus
\[\max_{Q} \mathbb{E}_{Q}\left[u\left(x\right)\right], \quad \text{subject to} \quad Q \in \mathcal{M}\left(P\right)\]
where $u\left(x\right)$ is the \textit{ex-post} utility of the sender from inducing posterior mean $x$. 

\cite{kam} use the Fenchel-Bunt extension (see e.g. Theorem 1.3.7 of \cite{ur}) of Carath\'eodory's theorem to show that an optimal signal in an $n$-state persuasion problem requires at most $n$ signal realizations. Our results above allow us to derive an analog of this result for linear persuasion problems with ease. 

By \cite{hill} the set $\mathcal{M}\left(P\right)$ is convex and closed. By the Riesz Representation theorem, the space of all probability measures with the same (compact) support is compact in the weak-* topology. Since $\mathcal{M}\left(P\right)$ is a closed subset of this space it must also be weak-* compact and therefore Hausdorff.

Bauer's Maximum Principle (see e.g. Corollary 7.70 of \cite{hitch}) implies that a linear functional such as the expectation operator on a (non-empty) compact and convex set $A$ of a locally-convex Hausdorff space attains its maximum at an extreme point of $A$. By Corollary \ref{useful} a necessary condition for $Q$ to be an extreme point of $\mathcal{M}\left(P\right)$ is that it have support on at most $n$ points. Thus, we have proved the following result.

\begin{proposition}\label{bound}
In linear persuasion problems in which the prior, $P$, has support on $n$ points, the optimal signal requires at most $n$ messages to be used.
\end{proposition}

\subsection{Simplifying Competitive Persuasion}\label{example}

Now consider a simple competitive persuasion problem as studied by \cite{Au2}. There are $k$ competing sellers who want a single risk-neutral buyer to purchase their good. The sellers sell products with quality given by an i.i.d. random variable $X_{i}$, $i = 1,...,k$, distributed according to the purely atomic measure $P$. 

Each seller, without knowing the quality of her (or the other sellers') goods, simultaneously chooses a Blackwell experiment conditioned on her type. The buyer observes the realization of the $k$ experiments and selects the seller whose product's posterior expected value is highest (and randomizes fairly when indifferent). A seller gets a payoff normalized to one if the buyer selects her and zero otherwise. Because it is only the expected value of the product that matters to the buyer (and because the sellers' payoffs are state-independent), this is a linear persuasion problem and so again each seller's problem is to choose a distribution of posterior means that is an mpc of the prior. That is, each seller's set of pure strategies is $\mathcal{M}\left(P\right)$.

As \cite{Au2} show, the unique symmetric equilibrium in pure strategies requires an uncountable number of signal realizations--each seller's equilibrium distribution, $Q$, is continuous except possibly at the upper bound of its support. However, such an equilibrium may be infeasible practically, either because the firms are unable to construct a signal of the necessary complexity or due to cognitive constraints suffered by the citizens. Corollary \ref{useful} allows us to sidestep this issue:

\begin{proposition}\label{mix}
For the $k$ player competitive persuasion game in which the prior distribution over qualities has support on $n$ points, there exists a symmetric mixed strategy equilibrium supported on mpcs with support on at most $n$ points.
\end{proposition}
\begin{proof}
Theorem 4 in \cite{Au2} establishes that a symmetric pure strategy equilibrium exists. Let $Q^{*}$ denote an equilibrium pure strategy distribution over (expected) qualities. By Corollary \ref{useful}, $Q^{*}$ can be obtained as a mixture of mpcs, each with support on at most $n$ points. Since this mixture results in the same distribution over qualities and the same payoff to each player as the pure strategy, $Q^{*}$, there must also be an equilibrium of the game in which each player mixes.
\end{proof}

As the next example illustrates, Corollary \ref{useful} can also be used to easily establish whether candidate pairs of strategies are equilibria.

\begin{example}
Consider the two player competitive persuasion problem with
\[P = \begin{Bmatrix} 
0 & \frac{1}{2} & \frac{3}{4}\\
\frac{1}{6} & \frac{1}{2} & \frac{1}{3}
\end{Bmatrix}\]
We claim that there is a Nash Equilibrium in which both sellers choose the cdf 
\[F(x) = \begin{cases}
\frac{2}{3}x, & \quad 0 \leq x \leq \frac{1}{2}\\
\frac{8}{3}x-1, & \quad \frac{1}{2} \leq x \leq \frac{3}{4}
\end{cases}\]
From Corollary \ref{useful} it suffices to check that there is no profitable deviation to an mpc with support on three points. To that end, suppose that seller $2$ chooses such a distribution, distribution $Q$:
\[Q = \begin{Bmatrix} 
a & b & c\\
p & q & r
\end{Bmatrix}\]
where $Q \in \mathcal{M}\left(P\right)$, $pa + qb + rc = 1/2$, and $p + q + r = 1$. Without loss of generality $a \leq 1/2$ and $c \geq 1/2$. First suppose that $b \leq 1/2$. Then, seller $2$'s payoff from deviating is 
\[\begin{split}
    u_{2} &= p F(a) + q F(b) + r F(c) = \frac{1}{3} + r\left(2c-1\right)
\end{split}\]
Because $c \leq 1/2 + 1/(12r)$,
\[\begin{split}
    u_{2} &\leq \frac{1}{3} + r\left(1+\frac{1}{6r}-1\right) = \frac{1}{2}
\end{split}\]
In a similar manner, for $b \geq 1/2$ we have
\[\begin{split}
    u_{2} &= p F(a) + q F(b) + r F(c) = \frac{1}{3} + p\left(1-2a\right)\\
\end{split}\]
where we use the fact that $p+q+r=1$ and that $pa + qb + rc = 1/2$. Then, since $a \geq (1-1/(6p))/2$
\[\begin{split}
    u_{2} &\leq \frac{1}{3} + p\left(1+\frac{1}{6p}-1\right) = \frac{1}{2}
\end{split}\]
Thus, there is no profitable deviation to any mpc with three point support and so no profitable deviation to any mpc.

\end{example}

\bibliography{sample.bib}
\end{document}